\documentclass[11pt]{article}
\usepackage{amsthm,amsmath,amssymb}
\usepackage[bookmarks=true,hypertexnames=false,pagebackref]{hyperref}
\hypersetup{colorlinks=true, citecolor=blue, linkcolor=red, urlcolor=blue}
\usepackage{newpxtext}
\usepackage[libertine,vvarbb]{newtxmath}

\usepackage[T1]{fontenc} 
\usepackage{textcomp} 
\usepackage[scr=rsfso]{mathalfa}
\usepackage{nicefrac}
\usepackage{cleveref}
\usepackage{graphicx}
\usepackage[ruled,vlined,linesnumbered]{algorithm2e}
\usepackage{aliascnt}
\usepackage[section]{placeins}
\usepackage{tcolorbox}
\usepackage{epigraph}
\usepackage{algpseudocode}
\usepackage{xifthen}
\usepackage{latexsym}
\usepackage{framed}
\usepackage{fullpage}
\usepackage{bm}
\usepackage{tikz}
\usetikzlibrary{positioning, fit, shapes.geometric, arrows}
\usepackage[normalem]{ulem}

\newtheorem{theorem}{Theorem}[section]
\newtheorem{claim}[theorem]{Claim}

\newtheorem{lemma}[theorem]{Lemma}
\newtheorem{conjecture}[theorem]{Conjecture}

\theoremstyle{definition}

\newtheorem{definition}[theorem]{Definition}

 \newcommand{\eps}{\varepsilon}
 
\renewcommand{\mid}{\;\middle\vert\;}

\newcommand{\rv}[1]{\boldsymbol{#1}}

\newcommand{\ent}{\mathbf{H}}
\renewcommand{\Pr}{\operatorname*{\mathbf{Pr}}}
\DeclareMathOperator*{\E}{\mathbf{E}}

\newcommand{\pr}[2][]{ \ifthenelse{\isempty{#1}}
  {\Pr\left[#2\right]} {\Pr_{#1}\left[#2\right]} }
\newcommand{\ex}[2][]{ \ifthenelse{\isempty{#1}}
  {\E\left[#2\right]}
  {\E_{#1}\left[#2\right]} }

\newcommand{\GHM}{\text{GHM}}

\newcommand{\delete}[1]{}
\usepackage[normalem]{ulem}
\usepackage{thm-restate}
\usepackage{enumerate}

\title{Quantum versus Classical Separation in  Simultaneous Number-on-Forehead Communication}

\begin{document}

\author{
Guangxu Yang \thanks{Thomas Lord Department of Computer Science, University of Southern California. Research supported by NSF CAREER award 2141536. Email: \texttt{\{guangxuy,jiapengz\}@usc.edu}}
\and
Jiapeng Zhang \footnotemark[1]
}
\maketitle

\begin{abstract}
Quantum versus classical separation plays a central role in understanding the advantages of quantum computation. In this paper, we present the first exponential separation between quantum and bounded-error randomized communication complexity in a variant of the Number-on-Forehead (NOF) model. Namely, the three-player \emph{Simultaneous Number-on-Forehead} model. Specifically, we introduce the \emph{Gadgeted Hidden Matching Problem} and show that it can be solved using only $O(\log n)$ simultaneous quantum communication. In contrast, any simultaneous randomized protocol requires $\Omega(n^{1/16})$ communication.

On the technical side, a key obstacle in separating quantum and classical communication in NOF models is that all known randomized NOF lower bound tools, such as the discrepancy method, typically apply to both randomized and quantum protocols. In this regard, our technique provides a new method for proving randomized lower bounds in the NOF setting and may be of independent interest beyond the separation result.

\end{abstract}
\section{Introduction}
One of the central goals in the study of quantum advantage is to demonstrate separations between quantum and classical computation in various computational models. In this direction, substantial progress has been achieved in communication complexity, where a line of work~\cite{buhrman1998quantum,raz1999exponential,bar2004exponential,gavinsky2007exponential,regev2011quantum,gavinsky2016entangled,girish2022quantum,gavinsky2019quantum,gavinsky2020bare,goos2024quantum} has established exponential separations in multiple settings. These results provide explicit problems that can be solved by quantum protocols using only $(\log n)^{O(1)}$ communication, whereas any classical randomized protocol solving the same problem must use $n^{\Omega(1)}$ communication.

However, all of the aforementioned results pertain to two-party communication models. In contrast, separations in multiparty communication, specifically the Number-on-Forehead (NOF) models, remain poorly understood. Lower bounds in the NOF models are especially intriguing because NOF protocols can simulate a broader range of computational models than those in the two-party setting. To this end, G{\"o}{\"o}s, Gur, Jain, and Li~\cite{goos2024quantum} highlighted the separation of quantum and randomized communication in NOF models as an important open problem.

The main obstacle to obtaining strong quantum-versus-classical separations in the NOF setting, as noted by~\cite{goos2024quantum}, lies in the lack of techniques for proving lower bounds for randomized NOF protocols. Existing approaches, such as the discrepancy method, often apply equally well to quantum communication, making it difficult to distinguish between the power of quantum and classical protocols in the NOF model~\cite{lee2009lower}.

In this paper, we study the \textit{simultaneous NOF model}, a variant of the NOF communication model. In this setting, each player \( i \) sends a single message, computed based on all inputs except \( x_i \), to a referee (or the last player), who then determines the output after receiving all the messages. Although the simultaneous NOF model appears weaker than the standard NOF model, proving strong lower bounds in this setting remains highly challenging. 

Strong simultaneous NOF lower bounds have significant implications in various areas, including lower bounds for the \textrm{ACC} circuit class~\cite{HG90,pudlak1997boolean}, private information retrieval~\cite{chor1998private}, and position-based cryptography~\cite{brody2017position}. Furthermore, many known non-trivial NOF protocols such as the $\sqrt{\log N}$ protocol for \textit{Exactly-N} \cite{chandra1983multi}, the $O\left(\frac{n\log\log n}{\log n}\right)$ for multipointer jumping \cite{brody2015dependent} and Shifting \cite{HG90},  and Grolmusz protocols for generalized inner-product \cite{grolmusz1994bns} and set disjointness \cite{rao2020communication} are all simultaneous protocols.

\subsection{Our Results}
In this paper, we establish the first exponential separation between quantum and bounded-error randomized communication complexity in the simultaneous NOF model. Our problem is inspired by the \emph{Hidden Matching Problem} introduced by~\cite{bar2004exponential}. In the Hidden Matching Problem, 
\begin{enumerate}
    \item Alice is given a string \( z \in \{0,1\}^n \). 
    \item Bob is given $M\in \mathcal{M}_n$ where $\mathcal{M}_n$ denotes the family of all possible perfect matchings on $n$ nodes. 
\end{enumerate} Their goal is for Bob to output a tuple \( \langle i, j, b \rangle \) such that the edge \( (i, j) \) belongs to \( M \) and \( b = z_i \oplus z_j \). As shown in~\cite{bar2004exponential}, the Hidden Matching Problem admits an \( O(\log n) \) simultaneous quantum communication protocol, while any one-way randomized protocol requires \( \Omega(\sqrt{n}) \) communication.

Inspired by \cite{bar2004exponential}, we introduce the \textit{Gadgeted Hidden Matching Problem} (GHM). Let \( m := n/2 \). For each \( i \in [m] \), we define a perfect matching \( M_i \) between $\{0, \dots, m-1\}$ and $\{m, \dots, 2m-1\}$ as
\[
M_i := \left\{ \left(\ell, m + \left((i + \ell) \bmod m\right)\right) : \ell \in [m] \right\}.
\]
Let \( \mathcal{M} = \{M_1, \dots, M_m\} \) denote the collection of these matchings. The GHM is defined as follows:

\begin{definition}
\label{dfn: GHM}
Let \( g : \{0,1\}^n \times \{0,1\}^n \rightarrow [m] \) be any gadget function. The \textit{Gadgeted Hidden Matching Problem}, denoted \( \GHM\circ g \), involves three inputs distributed among the players as follows:
\begin{itemize}
    \item Alice receives \( z,y \in \{0,1\}^n \),
    \item Bob receives \( z,x \in \{0,1\}^n \),
    \item Charlie receives \( x, y \in \{0,1\}^n\).
\end{itemize}
The goal is for Charlie to output a tuple \( \langle \ell, r, b \rangle \) such that \( (\ell, r) \in M_{g(x, y)} \) and \( b = z_\ell \oplus z_{r} \).
\end{definition}

Remarkably, this problem remains easy for quantum communication. Similar to the Hidden Matching Problem, Alice only needs to send a uniform superposition of the string \( z \), with a communication cost of \( O(\log n) \) qubits. Charlie can then perform a measurement on this superposition, which depends on the matching \( M_{g(x,y)} \), and output the parity of some pair in \( M_{g(x,y)} \) (see the appendix for more details). We note that Bob does not need to send any message using this protocol. Thus, the main result of our paper is the \( n^{\Omega(1)} \) randomized simultaneous NOF communication.

\begin{theorem}\label{main_theorem}
There exists an explicit gadget function \( g : \{0,1\}^n \times \{0,1\}^n \rightarrow [m] \) such that the randomized simultaneous NOF communication complexity of \( \GHM \circ g \) is \( \Omega(n^{1/16}) \).
\end{theorem}

\noindent
When \( g \) is clear from context, we simplify \( \GHM \circ g \) to \( \GHM \).

\section{Preliminaries}
We begin by fixing some notation. The set of integers $\{0, \ldots, n-1\}$ is denoted by $[n]$. We use capital letters like $X$ to denote sets, and bold symbols like $\rv{X}$ to denote random variables. In particular, for a set $X$, we use $\rv{X}$ to denote the random variable that is uniformly distributed over the set $X$.

A \textit{search problem} is a relation $S \subseteq X \times Y \times Z \times Q$, where $Q$ is the set of possible solutions. On input $(x, y, z) \in X \times Y \times Z$, the goal is to find a solution $q \in Q$ such that $(x, y, z, q) \in S$. Note that for the (Gadgeted) Hidden Matching Problem, all inputs are guaranteed to have at least one solution.

\subsection{Simultaneous Number-on-Forehead Model}  
In the three-party simultaneous NOF communication complexity model, Alice, Bob, and Charlie collaborate to compute a search problem \( S\subseteq X \times Y \times Z \times Q\). Their inputs are as follows:  
\begin{itemize}
    \item Alice receives \((y,z) \in Y \times Z\).
    \item Bob receives \((x,z) \in X \times Z\).
    \item Charlie receives \((x,y) \in X \times Y\).
\end{itemize}  

\noindent The randomized SM protocol \(\Pi = (\Pi_A, \Pi_B, \Pi_C)\) proceeds as follows:  
\begin{enumerate}
    \item Alice and Bob simultaneously send messages to Charlie, where Alice’s message \(\Pi_A(y,z,r)\) depends only on the input \((y,z)\) and public randomness $r$ and Bob’s message \(\Pi_B(x,z,r)\) depends only on the input \((x,z)\) and public randomness $r$.
    \item After receiving both messages, Charlie outputs a solution \(q= \Pi_C(\Pi_A(y,z), \Pi_B(x,z), x, y,r) \in Q\).
\end{enumerate}
The protocol $\Pi$ computes $S$ with error  $\epsilon$ if for any $(x,y,z)$, $\Pr_r[(x,y,z,q) \in S]\geq 1-\epsilon$.

\subsection{Basics of Information Theory}\label{sec:ic}
Our proof approach involves several standard definitions and results from information theory, which we now recall.

\begin{definition}[Entropy]
Given a random variable \(\boldsymbol{X}\), the Shannon entropy of \(\boldsymbol{X}\) is defined by 
\[
\ent(\boldsymbol{X}) := \sum_{x} \Pr(\boldsymbol{X} = x) \log \left( \frac{1}{\Pr(\boldsymbol{X} = x)} \right).
\]
For two random variables \(\boldsymbol{X}, \boldsymbol{Y}\), the \textit{conditional entropy} of \(\boldsymbol{X}\) given \(\boldsymbol{Y}\) is defined by 
\[
\ent(\boldsymbol{X} ~|~ \boldsymbol{Y}) := \mathbb{E}_{y \sim \boldsymbol{Y}} \left[ \ent(\boldsymbol{X} \mid \boldsymbol{Y} = y) \right].
\]
\end{definition}

\noindent For \(p \in [0,1]\), the binary entropy function is defined as $H_2(p) = - p \log_2 p - (1 - p) \log_2 (1 - p).$ It is well known that \(H_2(p)\) is a concave function.

\begin{lemma}[Subadditivity of Entropy]\label{subadditivity}
For a list of random variables \(\boldsymbol{X}_1, \boldsymbol{X}_2, \ldots, \boldsymbol{X}_d\), we have:
\[
\ent(\boldsymbol{X}_1, \boldsymbol{X}_2, \ldots, \boldsymbol{X}_d) \leq \sum_{i=1}^{d} \ent(\boldsymbol{X}_i).
\]
\end{lemma}

\begin{definition}[Mutual Information]
The mutual information between joint random variables $\boldsymbol{X}$ and $\boldsymbol{Y}$ is defined as
\[
\mathrm{I}(\boldsymbol{X}; \boldsymbol{Y}) = \ent(\boldsymbol{X}) - \ent(\boldsymbol{X} | \boldsymbol{Y}),
\]
\end{definition}

\begin{lemma}[Data Processing Inequality]\label{data_process}
Consider random variables \(\rv{X}, \rv{Y}, \rv{Z}\) forming a Markov chain \(\rv{X} \rightarrow \rv{Y} \rightarrow \rv{Z}\). Then, the mutual information satisfies:
\[
\mathrm{I}(X; Y) \geq \mathrm{I}(X; Z).
\]
\end{lemma}

\begin{definition}[Hamming Distance]
Let $x = (x_1, \ldots, x_n), y = (y_1, \ldots, y_n )\in \{0,1\}^n$ be two strings. Their \emph{Hamming distance} \(d_H(x, y)\) is defined as:
\[
d_H(x, y) := |\{i: x_i\neq y_i\}|.
\]
\end{definition}

\section{The Randomized Lower Bound}\label{sec:proof}

We prove the our main theorem in this section. We first recall the statement.

\begin{theorem}[Theorem \ref{main_theorem} restated]
There is a gadget $g$. Any randomized simultaneous NOF protocol that computes $\GHM\circ g$ with an error probability less
than $1/8$ requires $\Omega(n^{1/16})$ bits of communication.
\end{theorem}

To prove the randomized communication lower bound, we first describe the gadget function and a hard input distribution for the communication problem.

\subsection{The Gadgets and Hard Distributions}\label{gadget}
For $n>0$, we set $m=n/2$ to be a prime \footnote{Using a prime number when proving the lower bound is reasonable, because for every integer $n$ there exists a prime in the interval $[n/2,n]$.
} and $\alpha = \lfloor\sqrt{m}\rfloor$ in our proof.
\begin{definition}[Rich gadgets]
\label{gadget_function}
For $\alpha > 0$, we say that a gadget $g: \{0,1\}^n \times \{0,1\}^n \rightarrow [m]$ is \emph{$\alpha$-rich} if for every subset $Q \subseteq [m]$ with $|Q| = \alpha$, there exist sets $S, T \subseteq \{0,1\}^n$ with $|S| = |T| = \sqrt{\alpha}$ such that
\[
\{g(x,y) : x \in S, y \in T\} = Q.
\]
\end{definition}

\noindent It is not difficult to construct rich gadgets under our parameter choices.  
Note that the total number of subsets $Q \subseteq [m]$ of size $|Q| = \alpha$ is at most
\[
\binom{m}{\alpha} \leq m^{\sqrt{m}} = 2^{o(n)}.
\]
Therefore, we can assign disjoint sets $S_Q$ and $T_Q$ for each $Q$ and enforce that $g(S_Q, T_Q) = Q$. In the following proof, we fix the gadget $g$ to be any $\alpha$-rich gadget. We then define the hard distributions.

\begin{definition}
Given $S, T \subseteq \{0,1\}^n$ with $|S| = |T| = \sqrt{\alpha}$, we define the distribution $\mu(S,T)$ on $\{0,1\}^n \times \{0,1\}^n \times \{0,1\}^n$ as follows:
\begin{itemize}
    \item Uniformly sample $x\in S$, $y\in T$ and $z \in \{0,1\}^n$.
    \item Output the triple $(x, y, z)$.
\end{itemize}
\end{definition}

In our proof, we show that there exists a pair $(S, T)$ such that $\mu(S, T)$ is hard for any protocol. The choice of $S$ and $T$ depends on the gadget $g$.

\subsection{Local Independent Protocols}\label{symmetric}

To simplify our analysis, we first apply the following \emph{local independentization} process to any simultaneous NOF protocol.

\begin{definition}[Local independent protocols]
For fixed sets $S,T \subseteq \{0,1\}^n$ with $|S|=|T|=\sqrt{\alpha}$, and any protocol $\Pi$ for $\GHM \circ g$ under the distribution $\mu(S,T)$, we define its \emph{local independentized} version $\Pi^*$ as follows:
\begin{itemize}
    \item For any $(y,z) \in T \times Z$, $\Pi^{*}_A(y,z)$ outputs the tuple $\left( \Pi_A(u,z) \right)_{u \in T}$.
    \item For any $(x,z) \in S \times Z$, $\Pi^{*}_B(x,z)$ outputs the tuple $\left( \Pi_B(v,z) \right)_{v \in S}$.
    \item $\Pi_C^*(x,y,\Pi^{*}_A(y,z),\Pi^{*}_B(x,z))$ outputs the same value as $\Pi_C(x,y,\Pi_A(y,z),\Pi_B(x,z))$.
\end{itemize}
\end{definition}

The high-level intuition behind the local independentization process is that both Alice and Bob enumerate all their possible inputs and output all corresponding transcripts. We observe the following useful property of locally independent protocols.

\begin{claim}\label{symmetric_property}
Let $\Pi$ be a deterministic protocol for $\GHM \circ g$ with $\delta$-error under the distribution $\mu(S,T)$. Then the following statements hold:
\begin{enumerate}
    \item\label{1} The protocol $\Pi^*$ is a deterministic protocol for $\GHM \circ g$ with $\delta$-error under the distribution $\mu(S,T)$.
    \item\label{2} The communication complexity $\Pi^{*}$ is $\sqrt{\alpha} \cdot \mathrm{CC}(\Pi)$.
    \item\label{3} Under the distribution $\mu(S,T)$, the messages $\Pi_A^*(y,z)$ and $\Pi_B^*(x,z)$ depend only on $z$. That is, they are independent of $y$ and $x$, respectively.
\end{enumerate}
\end{claim}

\noindent The third item significantly simplifies our analysis. The proof of this claim is straightforward and is omitted here.

\subsection{Information Complexity of Local Independent Protocols}\label{encoding}

Now we are ready to prove the main result. Our proof is based on lower bounding the information complexity. However, instead of analyzing the standard information complexity $\mathrm{I}(\rv{Z} : \rv{\Pi})$, we lower bound the information complexity for the local independent protocols.

\begin{theorem}\label{Information_complexity}
Let $g$ be an $\alpha$-rich gadget. There exist sets $S, T \subseteq \{0,1\}^n$ with $|S| = |T| = \sqrt{\alpha}$ such that for any protocol $\Pi$ of $\GHM \circ g$ under the distribution $\mu(S,T)$ with error $1/16$, we have
\[
\mathrm{I}(\rv{Z} : \rv{\Pi^{*}}) = \Omega(m^{5/16}).
\]
\end{theorem}

We note that our main theorem (Theorem \ref{main_theorem}) is a direct consequence of Theorem \ref{Information_complexity}. By Theorem \ref{Information_complexity}, we know the communication complexity of $\Pi^{*}$ is at least $\Omega(m^{5/16})$, and hence the communication complexity of $\Pi$ is at least $\Omega(m^{5/16}/\sqrt{\alpha}) = \Omega(m^{1/16})$.

\begin{proof}[Proof of Theorem \ref{Information_complexity}]
The proof proceeds by carefully constructing the sets $S$ and $T$ based on the gadget function $g$ and the combinatorial structure provided by the following lemma.

\begin{lemma}
\label{bipartite}
There is a sequence \( Q \in [m]^{\alpha} \) consisting of distinct elements such that for any subsequence \( Q' = (q_1, \dots, q_r) \subseteq Q \) with \( r \geq \alpha/2 \), the following holds: for any sequence \( L = (\ell_1, \dots, \ell_r) \), either
\begin{itemize}
    \item \( L \) contains at least \( \Omega(m^{5/16}) \) distinct elements, or
    \item the set \( R := \left\{ m + ((\ell_i + q_i) \bmod m) ~|~ 1 \leq i \leq r \right\} \) satisfies \( |R| \geq \Omega(m^{5/16}) \).
\end{itemize}
\end{lemma}

We defer the proof of Lemma~\ref{bipartite} to Section~\ref{proof_structure} and first prove Theorem~\ref {Information_complexity} by assuming it. Let \(Q\in [m]^{\alpha} \) be the sequence guaranteed by the Lemma \ref{bipartite}. Since $g$ is a $\alpha$-rich gadget, there exist subsets $S, T \subseteq \{0,1\}^n$ with $|S| = |T| = \sqrt{\alpha}$ such that,
\[
\{g(x,y) : x \in S, y \in T\} = Q.
\]

Now we show that this pair $S,T$ is the desired set that satisfies Theorem~\ref{Information_complexity}.
Let $\Pi$ be a deterministic protocol for $\GHM \circ g$ under $\mu(S,T)$, and let $\Pi^*$ be the local independentized version of it.
Our goal is to prove a lower bound on the mutual information $\mathrm{I}(\rv{Z} : \rv{\Pi^{*}})$. Recall that the output of $\Pi^*$ has the form 
\[
\Pi^*(x,y,z) = (\Pi_A^*(z), \Pi_B^*(z), \Pi_C^*(x,y,\Pi_A^*(z),\Pi_B^*(z))).
\]
Here $\Pi_A^*(z)$, $\Pi_B^*(z)$ depend only on $z$ as $\Pi^*$ is locally independent.
The output of $\Pi_C^*$ is a triple $(\ell,r,b)$, where $\ell,r\in [n]$ and $b \in \{0,1\}$.
The output is correct if $(\ell,r)$ is an edge in $M_{g(x,y)}$ and $b = z_\ell \oplus z_r$.

Since Charlie knows the matching $M_{g(x,y)}$ and outputs $\Pi_C^*$, we may assume without loss of generality that $(\ell, r)$ is always an edge in $M_{g(x,y)}$.  
Hence, errors occur only if $b \neq z_\ell \oplus z_r$. The error probability of $\Pi^*$ under the input distribution $\mu(S,T)$ can therefore be written as
\[
\mathcal{E}_{\Pi^{*}}(\rv{S}, \rv{T}, \rv{Z}) = \Pr_{(x,y,z) \sim \mu(S,T)}[\Pi^{*}(x,y,z) \notin \GHM(x,y,z)] = \Pr_{(x,y,z) \sim \mu(S,T)}[b \neq z_\ell \oplus z_r] .
\]
By Claim~\ref{symmetric_property}, we have $\mathcal{E}_{\Pi^{*}}(\rv{S}, \rv{T}, \rv{Z}) \leq 1/16$. Since the messages from Alice and Bob now depend only on $z$, their combined message induces a partition of $Z$. For each message $\tau$, we define
\[
Z_\tau = \{ z ~|~ (\Pi_A^*(z), \Pi_B^*(z)) = \tau \}
\]
and let $\rv{Z}_\tau$ be uniformly distributed over $Z_\tau$. Define the conditional error as
\[
e_{\tau} := \mathcal{E}_{\Pi^{*}}(\rv{S}, \rv{T}, \rv{Z}_{\tau}) = \Pr_{(x,y,z) \sim \mu(S,T)}[\Pi^{*}(x,y,z) \notin \GHM(x,y,z) ~|~ z \in Z_{\tau}] .
\]
Notice that $\mathbb{E}[e_\tau] = \mathcal{E}_{\Pi^{*}}(\rv{S}, \rv{T}, \rv{Z}) \leq 1/16$, then by Markov’s inequality we have 
\[
\Pr[e_\tau \geq 1/8] \leq 1/2.
\]
The following lemma shows that any message $\tau$ with small error reveals a lot of information.

\begin{lemma}\label{entropy_loss}
For every message $\tau$ with $e_{\tau} < 1/8$, we have:
\[
\ent(\rv{Z}) - \ent(\rv{Z} | (\Pi_A^*(z), \Pi_B^*(z))  = \tau) = \Omega(m^{5/16}).
\]
\end{lemma}
\noindent By assuming Lemma \ref{entropy_loss} and using the fact that $\Pr[e_{\tau} < 1/8]  \geq 1/2$, we have
\[
\mathrm{I}(\rv{Z}:\rv{\Pi^{*}}) = \Omega(m^{5/16}).
\]
\end{proof}
\noindent Now we focus on the proof of Lemma \ref{entropy_loss}.
\begin{proof}[Proof of Lemma \ref{entropy_loss}]
For a fixed message $\tau$ with $e_{\tau}<1/8$. Recall by Definition \ref{dfn: GHM} that the problem $\GHM\circ g$ is defined on a set of perfect matching $\{M_1,\dots,M_{m}\}$ where 
\[
M_i := \left\{ \left(\ell, m + \left((i + \ell) \bmod m\right)\right) : \ell \in [m] \right\}.
\]
As $S$ and $T$ are now fixed, we are specifically interested in those matching:
\[
\mathcal{M} = \{M_i:\exists x\in S, y\in T, g(x,y) = i\}
\]
Furthermore, since $g$ is $\alpha$-rich, for each $M_i\in \mathcal{M}$, there is a unique pair $x\in S$ and $y\in T$ such that $g(x,y) = i$. We denote it by $(x,y)=g^{-1}(i)$. For each $M_i \in \mathcal{M}$, let 
\[
e_{\tau,M_i} = \mathcal{E}_{\Pi^{*}}(\rv{Z}_{\tau}, M_i):= \Pr_{(x,y,z)\sim\mu(S,T)}\left[\Pi^{*}(x,y,z) \notin \GHM(x,y,z)~|~z\in Z_{\tau}, M_{g(x,y)}=M_i \right]
\] 
Recall that $\mathbb{E}_{M_i}[e_{\tau, M_i}] = e_{\tau} < 1/8$. Then by Markov's inequality again, the set 
\[
\mathcal{M}' := \{ M_i \in \mathcal{M} : e_{\tau, M_i} \leq 1/4 \}
\]
has size at least $|\mathcal{M}'| \geq |\mathcal{M}|/2 = \alpha/2$.

For every $M_i \in \mathcal{M}'$, let $(x, y)=g^{-1}(i)$, and let $(\ell_i, r_i, b_i)$ be the tuple output by $\Pi_{C}^{*}(x, y, \tau)$.  
Recall that Charlie always outputs a pair $(\ell_i, r_i)$ that belongs to the matching defined by $(x, y)$. Hence, it must have the form
\[
r_i = m + \left((\ell_i + i) \bmod m\right).
\]
Let $G_\tau = (L, R, E)$ be a bipartite graph with vertices $L = \{0, 1, \dots, m-1\}$ and $R = \{m, m+1, \dots, 2m-1\}$. An edge connects $\ell \in L$ and $r \in R$ if and only if there exists $(x, y)$ with $M_{g(x, y)} \in \mathcal{M}'$ such that
\[
(\ell, r) = \Pi_{C}^{*}(x, y, \tau).
\]

By applying Lemma~\ref{bipartite} with $Q' := \{i : M_i \in \mathcal{M}'\}$, we conclude that either at least $\Omega(m^{5/16})$ vertices on the left side of $G_\tau$ are incident to at least one edge, or at least $\Omega(m^{5/16})$ vertices on the right side are incident to at least one edge.

In the first case, let $A \subseteq E$ be a set of edges such that each vertex on the left side is incident to exactly one edge in $A$. In the second case, choose $A$ such that each vertex on the right side is incident to at most one edge in $A$. In both cases, we have $|A| = \Omega(m^{5/16})$.

Recall that the set of edges $E$ is identical to $\mathcal{M}'$, so $A$ is indeed a subset of $\mathcal{M}'$. Let $v \in \{0,1\}^{A}$ be the vector defined by
\[
v_i = b_i,
\]
where $(\ell_i, r_i, b_i)$ is the output of $\Pi_{C}^{*}(g^{-1}(i), \tau)$ for each $i \in A$.

On the other hand, for each $z \in Z_{\tau}$, define a vector $u_z \in \{0,1\}^{A}$ by
\[
(u_z)_i = z_{\ell_i} \oplus z_{r_i}.
\]
Recall that $(u_z)_i$ corresponds to the correct answer of $\GHM \circ g$ on input $(\ell_i, r_i, z)$. Hence,
\[
\mathcal{E}_{\Pi^{*}}(\rv{Z}_\tau, \mathbf{A}) =
\Pr_{(x,y,z)\sim\mu(S,T)}\left[(u_z)_i \neq v_i ~|~ M_i \in A, z \in Z_\tau\right]
= \frac{\mathbb{E}[d_H(\rv{u}_z, v)]}{|A|} \leq 2 e_{\tau}\leq 1/4,
\]
where $\rv{u}_z$ denotes the distribution induced by sampling $z \sim \rv{Z}_\tau$ and computing $u_z$.

Next, we apply the following lemma, a generalization of Fano’s inequality due to Bar-Yossef, Jayram, and Kerenidis~\cite{bar2004exponential}, to upper bound the entropy of $\rv{u}_z$.

\begin{lemma}[\cite{bar2004exponential}]\label{fano}
Let $\rv{W}$ be a random variable over $\{0,1\}^k$, and suppose there exists a fixed vector $v \in \{0,1\}^k$ such that $\mathbb{E}[d_H(\rv{W}, v)] \leq \eps \cdot k$ for some $0 \leq \eps \leq 1/2$. Then the entropy of $\rv{W}$ is bounded by,
\[
\ent(\rv{W}) \leq k \cdot H_2(\eps),
\]
where $H_2(\cdot)$ is the binary entropy function.
\end{lemma}

\begin{proof}
Without loss of generality, we assume that $v = 0^k$. Let $\rv{W} = (\rv{W}_1, \dots, \rv{W}_k)$, where each $\rv{W}_i$ is a Bernoulli random variable with $p_i := \Pr[\rv{W}_i = 1]$. Then
\[
\mathbb{E}[d_H(\rv{W}, 0^k)] = \sum_{i=1}^k p_i \leq \eps \cdot k.
\]
By subadditivity of entropy and the concavity of $H_2(\cdot)$, we have
\[
\ent(\rv{W}) \leq \sum_{i=1}^k \ent(\rv{W}_i) = \sum_{i=1}^k H_2(p_i)
\leq k \cdot H_2\left( \frac{1}{k} \sum_{i=1}^k p_i \right)
\leq k \cdot H_2(\eps).
\]
\end{proof}

\noindent By applying this lemma to $\rv{u}_z$ with $\eps = 1/4$ and $k = |A|$, we obtain
$
\ent(\rv{u}_z) \leq |A| \cdot H_2(1/4).
$ Therefore,
\begin{align*}
\ent(\rv{Z} ~|~ (\Pi_A^*(z), \Pi_B^*(z)) = \tau)
&= \ent(\rv{u}_z ~|~ (\Pi_A^*(z), \Pi_B^*(z)) = \tau) + \ent(\rv{Z} ~|~ \rv{u}_z, (\Pi_A^*(z), \Pi_B^*(z)) = \tau) \\
&\leq |A| \cdot H_2(1/4) + \ent(\rv{Z} ~|~ \rv{u}_z, (\Pi_A^*(z), \Pi_B^*(z)) = \tau).
\end{align*}

Finally, since each vertex (on either the left or right side) in the matching is incident to at most one edge in $A$, we conclude that
\[
\ent(\rv{Z} ~|~ \rv{u}_z, (\Pi_A^*(z), \Pi_B^*(z)) = \tau) \leq \ent(\rv{Z}) - |A|.
\]
Putting everything together,
\[
\ent(\rv{Z} ~|~ (\Pi_A^*(z), \Pi_B^*(z)) = \tau) \leq \ent(\rv{Z}) - |A| (1 - H_2(1/4)).
\]
Since $|A| = \Omega(m^{5/16})$ and $1 - H_2(1/4) > 0$, this yields the desired entropy loss.

\end{proof}

\subsubsection{Proof of Lemma \ref{bipartite}}\label{proof_structure}
We prove Lemma~\ref{bipartite} in this section. Recall that $m$ is a prime number and $\alpha =\sqrt{m}$. This is a purely combinatorial problem, and we first recall the statement.

\begin{lemma}[Restate of Lemma~\ref{bipartite}]
There exists a sequence \( Q \in [m]^{\alpha} \) consisting of distinct elements such that for any subsequence \( Q' = (q_1, \dots, q_r) \subseteq Q \) with \( r \geq \alpha/2 \), the following holds: for any sequence \( L = (\ell_1, \dots, \ell_r) \), either
\begin{itemize}
    \item \( L \) contains at least \( \Omega(m^{5/16}) \) distinct elements, or
    \item the set \( R := \left\{ m + ((\ell_i + q_i) \bmod m) ~|~ 1 \leq i \leq r \right\} \) satisfies \( |R| \geq \Omega(m^{5/16}) \).
\end{itemize}
Here, $\alpha = \lfloor\sqrt{m}\rfloor$ is the parameter we chose previously.
\end{lemma}

The following periodic definition is the core concept in our proof of this lemma. 

\begin{definition}\label{dfn: periodic}
We view the sequence \(Q \in [m]^{\alpha}\) as a set of integers modulo \(m\). 
A subset \(A \subseteq Q\) is called \emph{periodic} if there exists a nonzero \(b \in [m] \setminus \{0\}\) such that
\[
A + b := \{ (a+b) \bmod m ~|~ a \in A \} \subseteq Q.
\]
We say that \(Q\) is \(\beta\)-periodic-free if no periodic subset \(A\subseteq Q\) with size \(|A| \geq \beta\).
\end{definition}

\noindent In our proof, we specifically choose \(\beta = 8\log m\).

\begin{lemma}\label{periodic}
Let \(\beta = 8\log m\). Then there exists a subset \(Q \subseteq [m]\) of size \(|Q| = \alpha\) that is \(\beta\)-periodic-free.
\end{lemma}
\begin{proof}
We use a probabilistic argument to prove that, i.e., we choose $Q$ uniformly at random with \(|Q| = \alpha\).
Fix \(b\in [m]\setminus\{0\}\), define
\[
A_b=\bigl\{\,a\in[m]\;\bigl|\;a\in Q\ \text{and}\ (a+b)\bmod m\in Q\bigr\}.
\]
and we aim to bound
\[
\Pr_{Q}\bigl[\exists b\in[m]\setminus\{0\}, |A_b|\geq\beta\,\bigr]
\]

\noindent First, we use the following claim to simplify our proof.
\begin{claim}\label{cl:mul-cycle}
Let \(m\) be a prime, for any  \(b\in[m]\setminus\{0\}\) and $a\in [m]$, 
\[
\bigl\{a+k\,b \bmod m\;\bigm|\;k=0,1,\dots ,m-1\bigr\} \;=\;[m].
\]
\end{claim}
\noindent By Claim~\ref{cl:mul-cycle}, we have that 
$
\Pr_{Q}\left[\,|A_b|\geq\beta\,\right] = \Pr_{Q}\left[\,|A_1|\geq\beta\,\right]
$
for any $b \in [m] \setminus \{0\}$. Hence, it suffices to analyze the case $b = 1$. We partition $[m]$ into three residue classes, defining
\[
S_{\gamma} := \left\{\, x \in [m] ~\middle|~ x \equiv \gamma \pmod{3} \,\right\} \quad \text{for } \gamma \in \{0,1,2\}.
\]
For each $\gamma$ and a set $Q \subseteq [m]$, we define:
\[
P_\gamma=\{(a,a+1)~|~ a\in S_\gamma\},\qquad
X_{a}(Q):=\mathbf 1(a\in Q,\,a+1\in Q),\qquad
Y_\gamma(Q):=\sum_{(a,a+1)\in P_\gamma}X_{a}(Q).
\]

\noindent Clearly, we have $|A_1| = Y_0(Q) + Y_1(Q) + Y_2(Q).$ Therefore, the desired probability becomes
\[
\Pr_{Q}\left[\,|A_1| \geq \beta\,\right] = \Pr_{Q}\left[\, Y_0(Q) + Y_1(Q) + Y_2(Q) \geq \beta\,\right].
\]
In what follows, we prove that for every $\gamma \in \{0,1,2\}$.
\[
\Pr\left[Y_\gamma(Q) \ge \frac{8}{3} \cdot \log m \right] \le \exp\left(-2(\log m)^2\right) \tag{$\ast$}
\]
Instead of sampling exactly $|Q| = \alpha$ elements, we alternatively consider the Bernoulli distribution $Q^*$, where each element $x \in [m]$ is included in $Q^*$ independently with probability $\alpha / m$. Observe that all pairs in $P_\gamma$ are pairwise disjoint. By the negative association of random variables \cite{joag1983negative}, we then have that 
\[
\Pr_{Q}\bigl[Y_\gamma(Q)\ge t\bigr] \le \Pr_{Q^{*}}\bigl[Y_\gamma(Q^*)\ge t\bigr].
\]

Under the Bernoulli distribution, $Y_\gamma(Q^*) = \sum_{(a,a+1)\in P_\gamma} X_a(Q^*)$ is a sum of $|P_\gamma| = \lfloor m/3 \rfloor$ independent indicator variables. For each $a$, we have
\[
\Pr[X_a = 1] = \Pr[a \in Q^*, a+1 \in Q^*] = \left(\frac{\alpha}{m}\right)^2 = \frac{1}{m},
\]
and hence the expected value $\mu := \mathbb{E}[Y_\gamma(Q^*)] \le 1/3$. We apply Chernoff’s inequality:
\[
\Pr\left[Y_\gamma(Q^*) \ge \frac{8\log m}{3} \right] 
= \Pr\left[Y_\gamma(Q^*) \ge (1+\delta)\mu \right]
\le \left(\frac{e^\delta}{(1+\delta)^{1+\delta}}\right)^\mu,
\]
where
\[
\delta := \frac{\frac{8\log m}{3}}{\mu} - 1 \ge 8 \log m - 1.
\]
Since $\mu \le 1/3$, it follows that
\[
\Pr_{Q}[Y_\gamma(Q) \ge \tfrac{8}{3} \log m] \le \Pr[Y_\gamma(Q^*) \ge \tfrac{8}{3} \log m] \le \exp\left(-2(\log m)^2\right). 
\]

\noindent Thus, by $(\ast)$:
\[
\Pr\left[ |A_1| \ge 8 \log m \right] 
= \Pr\left[ \sum_{\gamma = 0}^{2} Y_\gamma(Q) \ge 8 \log m \right]
\le \sum_{\gamma=0}^{2} \Pr\left[ Y_\gamma(Q) \ge \tfrac{8}{3} \log m \right]
\le 3 \exp\left(-2(\log m)^2\right).
\]

\noindent Recall that $\beta = 8 \log m$, applying a union bound over all $b \in [m] \setminus \{0\}$ gives
\[
\Pr_{Q: |Q| = \sqrt{m}} \left[ \exists\, b \in [m] \setminus \{0\} \text{ such that } |A_b| \ge \beta \right]
\le (m - 1) \cdot \frac{1}{2m} < \frac{1}{2}.
\]

\noindent Hence, there exists a subset $Q \subseteq [m]$ of size $|Q| = \alpha$ that is $\beta$-periodic-free.

\end{proof}

\noindent Now we are ready to prove Lemma~\ref{bipartite}.
\begin{proof}[Proof of Lemma \ref{bipartite}]
Fix any \(Q=(q_1,...,q_{\alpha})\subseteq [m]\) with $|Q| = \alpha$ that is \emph{\(\beta\)-periodic-free} by  Lemma \ref{periodic}. Since any subset of a \emph{\(\beta\)-periodic-free} set is also \emph{\(\beta\)-periodic-free}, any subsequence \(Q'\subseteq Q\) of size \(|Q'| = \alpha/2\) remains \emph{\(\beta\)-periodic-free}.

Let $L = (\ell_1, \dots, \ell_r)$ be any sequence.  
If $L$ contains at least $\Omega(m^{5/16})$ distinct elements, then the lemma follows immediately.  
Otherwise, let $t$ denote the number of distinct elements in $L$, and let these be denoted by $c_1, \dots, c_t$.  
We partition the set $R$ as follows:
\[
R = R_1 \cup \dots \cup R_t, \quad \text{where } R_j = \left\{ m + ((\ell_i + q_i) \bmod m) ~\middle|~ \ell_i = c_j \right\}.
\]
Without loss of generality, we assume that $|R_1| \geq |R_2| \geq \dots \geq |R_t|$. Since elements of \(Q'\) are distinct, we have
\[
|R_1|+\cdots+|R_t| = |Q'| \geq \alpha/2.
\]
We claim that for any \(i,j\in [t]\), it holds that
\[
|R_i\cap R_j|\leq \beta.
\]
Otherwise, if \(|R_i\cap R_j|>\beta\), consider \(A=R_i\cap R_j\). Then both \(A-c_i-m\) and \(A-c_j-m\) are subsets of \(Q'\), and since \(b=c_i-c_j\neq 0\), the set \(A-c_i-m\) would be a periodic subset of size larger than \(\beta\), contradicting the fact that $Q'$ is \emph{\(\beta\)-periodic-free}. Thus, by inclusion-exclusion,
\[
|R_1\cup \dots\cup R_j| \geq |R_1| + |R_2| - |R_2\cap R_1| +\cdots + |R_j| - |R_j\cap R_{1}|-\cdots -|R_j\cap R_{j-1}| > |R_1|+\cdots+|R_j| - \beta\cdot j^2.
\]
Since \(|R_1|\geq \cdots\geq |R_t|\), \(t\leq m^{5/16}\) and $\alpha = \sqrt{m}$, we have
\[
|R_1\cup \dots\cup R_j| \geq \frac{j\alpha}{2t} - \beta j^2 \geq j\cdot m^{3/16}/2 - \beta j^2.
\]

\noindent Choosing \(j=m^{2.2/16}\), since $\beta = 8\cdot \log m$ we obtain
\[
|R|\geq m^{5.2/16}/2 - 8\log m \cdot m^{4.1/16} > m^{5/16}
\]
for sufficiently large \(m\).
\end{proof}

\section{Open Problems}

We have shown an $\Omega(n^{1/16})$ randomized lower bound for the simultaneous NOF communication complexity of $\GHM \circ g$. On the other hand, the best known upper bound remains $\sqrt{n}$, where Alice simply sends $\sqrt{n}$ random bits of $z$. By the birthday paradox, with high probability, Charlie can recover the value of at least one matched pair from Alice’s message. Bridging this gap remains a compelling open question.

\begin{conjecture}
There exists a gadget $g$ such that the randomized simultaneous NOF communication complexity of $\GHM \circ g$ is $\Omega(\sqrt{n})$.
\end{conjecture}

Our techniques are currently tailored to the simultaneous NOF model. It remains open whether similar quantum-classical separations can be established in more general communication models, such as the one-way NOF model.

\begin{conjecture}
There exists a gadget $g$ such that the randomized one-way NOF communication complexity of $\GHM \circ g$ is $\Omega(n^{\Omega(1)})$.
\end{conjecture}

\bibliographystyle{alpha}
\bibliography{reference.bib}

\newcommand{\etalchar}[1]{$^{#1}$}
\begin{thebibliography}{GKK{\etalchar{+}}07}

\bibitem[BCW98]{buhrman1998quantum}
Harry Buhrman, Richard Cleve, and Avi Wigderson.
\newblock Quantum vs. classical communication and computation.
\newblock In {\em Proceedings of the thirtieth annual ACM symposium on Theory of computing}, pages 63--68, 1998.

\bibitem[BDFP17]{brody2017position}
Joshua Brody, Stefan Dziembowski, Sebastian Faust, and Krzysztof Pietrzak.
\newblock Position-based cryptography and multiparty communication complexity.
\newblock In {\em Theory of Cryptography Conference}, pages 56--81. Springer, 2017.

\bibitem[BS15]{brody2015dependent}
Joshua Brody and Mario Sanchez.
\newblock Dependent random graphs and multiparty pointer jumping.
\newblock {\em arXiv preprint arXiv:1506.01083}, 2015.

\bibitem[BYJK04]{bar2004exponential}
Ziv Bar-Yossef, Thathachar~S Jayram, and Iordanis Kerenidis.
\newblock Exponential separation of quantum and classical one-way communication complexity.
\newblock In {\em Proceedings of the thirty-sixth annual ACM symposium on Theory of computing}, pages 128--137, 2004.

\bibitem[CFL83]{chandra1983multi}
Ashok~K Chandra, Merrick~L Furst, and Richard~J Lipton.
\newblock Multi-party protocols.
\newblock In {\em Proceedings of the fifteenth annual ACM symposium on Theory of computing}, pages 94--99, 1983.

\bibitem[CKGS98]{chor1998private}
Benny Chor, Eyal Kushilevitz, Oded Goldreich, and Madhu Sudan.
\newblock Private information retrieval.
\newblock {\em Journal of the ACM (JACM)}, 45(6):965--981, 1998.

\bibitem[Gav16]{gavinsky2016entangled}
Dmitry Gavinsky.
\newblock Entangled simultaneity versus classical interactivity in communication complexity.
\newblock In {\em Proceedings of the forty-eighth annual ACM symposium on Theory of Computing}, pages 877--884, 2016.

\bibitem[Gav19]{gavinsky2019quantum}
Dmitry Gavinsky.
\newblock Quantum versus classical simultaneity in communication complexity.
\newblock {\em IEEE Transactions on Information Theory}, 65(10):6466--6483, 2019.

\bibitem[Gav20]{gavinsky2020bare}
Dmitry Gavinsky.
\newblock Bare quantum simultaneity versus classical interactivity in communication complexity.
\newblock In {\em Proceedings of the 52nd Annual ACM SIGACT Symposium on Theory of Computing}, pages 401--411, 2020.

\bibitem[GGJL24]{goos2024quantum}
Mika G{\"o}{\"o}s, Tom Gur, Siddhartha Jain, and Jiawei Li.
\newblock Quantum communication advantage in tfnp.
\newblock {\em arXiv preprint arXiv:2411.03296}, 2024.

\bibitem[GKK{\etalchar{+}}07]{gavinsky2007exponential}
Dmitry Gavinsky, Julia Kempe, Iordanis Kerenidis, Ran Raz, and Ronald De~Wolf.
\newblock Exponential separations for one-way quantum communication complexity, with applications to cryptography.
\newblock In {\em Proceedings of the thirty-ninth annual ACM symposium on Theory of computing}, pages 516--525, 2007.

\bibitem[Gro94]{grolmusz1994bns}
Vince Grolmusz.
\newblock The bns lower-bound for multiparty protocols is nearly optimal.
\newblock {\em Information and computation}, 112(1):51--54, 1994.

\bibitem[GRT22]{girish2022quantum}
Uma Girish, Ran Raz, and Avishay Tal.
\newblock Quantum versus randomized communication complexity, with efficient players.
\newblock {\em computational complexity}, 31(2):17, 2022.

\bibitem[HG90]{HG90}
J.~Hastad and M.~Goldmann.
\newblock On the power of small-depth threshold circuits.
\newblock In {\em Proceedings [1990] 31st Annual Symposium on Foundations of Computer Science}, pages 610--618 vol.2, 1990.

\bibitem[JDP83]{joag1983negative}
Kumar Joag-Dev and Frank Proschan.
\newblock Negative association of random variables with applications.
\newblock {\em The Annals of Statistics}, pages 286--295, 1983.

\bibitem[LSS09]{lee2009lower}
Troy Lee, Gideon Schechtman, and Adi Shraibman.
\newblock Lower bounds on quantum multiparty communication complexity.
\newblock In {\em 2009 24th Annual IEEE Conference on Computational Complexity}, pages 254--262. IEEE, 2009.

\bibitem[PRS97]{pudlak1997boolean}
Pavel Pudl{\'a}k, Vojtech R{\"o}dl, and Jir{\'\i} Sgall.
\newblock Boolean circuits, tensor ranks, and communication complexity.
\newblock {\em SIAM Journal on Computing}, 26(3):605--633, 1997.

\bibitem[Raz99]{raz1999exponential}
Ran Raz.
\newblock Exponential separation of quantum and classical communication complexity.
\newblock In {\em Proceedings of the thirty-first annual ACM symposium on Theory of computing}, pages 358--367, 1999.

\bibitem[RK11]{regev2011quantum}
Oded Regev and Bo'az Klartag.
\newblock Quantum one-way communication can be exponentially stronger than classical communication.
\newblock In {\em Proceedings of the forty-third annual ACM symposium on Theory of computing}, pages 31--40, 2011.

\bibitem[RY20]{rao2020communication}
Anup Rao and Amir Yehudayoff.
\newblock {\em Communication complexity: and applications}.
\newblock Cambridge University Press, 2020.

\end{thebibliography}
\newpage
\section*{Appendix}

\paragraph{Quantum protocols for $\textbf{GHM}_n$ \cite{bar2004exponential}:}
We present a quantum protocol for the gadgeted hidden matching problem with communication complexity of $O(\log n)$ qubits. Let $z=(z_1,\cdots,z_n) \in \{0,1\}^n$ and $y \in [n]$ be Alice’s input and $x,y \in [n]$ be Charlie’s input.

\begin{enumerate}
    \item Alice sends the state \( |\psi\rangle = \frac{1}{\sqrt{n}} \sum_{i=1}^{n} (-1)^{z_i} |i\rangle. \)
    
    \item Charlie performs a measurement on the state \( |\psi\rangle \) in the orthonormal basis
    \[
    B = \left\{ \frac{1}{\sqrt{2}} (|k\rangle \pm |\ell\rangle) \mid (k, \ell) \in M_{g(x,y)} \right\}.
    \]
\end{enumerate}

The probability that the outcome of the measurement is a basis state \( \frac{1}{\sqrt{2}} (|k\rangle + |\ell\rangle) \) is
\[
|\langle \psi | \frac{1}{\sqrt{2}} (|k\rangle + |\ell\rangle) \rangle|^2 = \frac{1}{2n} \left( (-1)^{z_k} + (-1)^{z_\ell} \right)^2.
\]

This equals $2/n$ if $z_k \oplus z_\ell = 0$ and $0$ otherwise. Similarly, for the states \( \frac{1}{\sqrt{2}} (|k\rangle - |\ell\rangle) \), we have that
\[
|\langle \psi | \frac{1}{\sqrt{2}} (|k\rangle - |\ell\rangle) \rangle|^2 = 0 \quad \text{if } z_k \oplus z_\ell = 0, \text{ and } \frac{2}{n} \text{ if } z_k \oplus z_\ell = 1.
\]

Hence, if the outcome of the measurement is a state \( \frac{1}{\sqrt{2}} (|k\rangle + |\ell\rangle) \), then Charlie knows with certainty that $z_k \oplus z_\ell = 0$ and outputs \( \langle k, \ell, 0 \rangle \). If the outcome is a state \( \frac{1}{\sqrt{2}} (|k\rangle - |\ell\rangle) \), then Charlie knows with certainty that $z_k \oplus z_\ell = 1$ and hence outputs \( \langle k, \ell, 1 \rangle \). Note that the measurement depends only on Charlie’s input and that the algorithm is $0$-error.

\end{document}